\documentclass[letterpaper, 10 pt, conference]{ieeeconf}  

\IEEEoverridecommandlockouts                              

\usepackage{amsmath} 
\usepackage{amssymb}  
\usepackage{amsfonts}
\usepackage{mathtools}
\usepackage{bbm}
\usepackage{bm}
\usepackage{color}
\usepackage{graphics} 
\usepackage{graphicx}
\usepackage{epsfig}   
\usepackage{epstopdf}
\usepackage[noadjust]{cite} 
\usepackage{tikz}
\usetikzlibrary{calc,positioning,shapes,shadows,arrows,fit}
\usepackage{fancyhdr}
\usepackage{framed}
\usepackage{float}
\usepackage{comment}
\usepackage{enumerate}
\usepackage{subfigure}
\usepackage[font = small]{caption}

\usepackage{amsthm}

\newtheorem{theorem}{Theorem}

\theoremstyle{definition}

\newtheorem{definition}{Definition}
\newtheorem{remark}{Remark}
\newtheorem{assumption}{Assumption}

\newtheorem{problem}{Problem}
 
\usepackage{algorithm}
\usepackage{algpseudocode}
\usepackage{pifont}

\renewcommand{\baselinestretch}{0.97}

\usepackage{caption}

\title{\LARGE \bf
Position and Orientation Based Formation Control of Multiple Rigid Bodies with Collision Avoidance and Connectivity Maintenance}

\author{Christos K. Verginis, Alexandros Nikou and Dimos V. Dimarogonas
\thanks{The authors are with the ACCESS Linnaeus Center, School of Electrical
Engineering, KTH Royal Institute of Technology, SE-100 44, Stockholm,
Sweden and with the KTH Center for Autonomous Systems. Email: {\tt\small \{cverginis, anikou, dimos\}@kth.se}. This work was supported by the H2020 ERC Starting Grant BUCOPHSYS, the Swedish Research Council (VR), the Swedish Foundation for Strategic Research (SSF), the Knut och Alice Wallenberg Foundation, the European Union's Horizon 2020 Research and Innovation Programme under the Grant Agreement No. 644128 (AEROWORKS) and the EU H2020 Research and Innovation Programme under GA No. 731869 (Co4Robots).}
}

\begin{document}

\maketitle
\thispagestyle{empty}
\pagestyle{empty}

\begin{abstract}
	This paper addresses the problem of position- and orientation-based formation control of a class of second-order nonlinear multi-agent systems in a $3$D workspace with obstacles. More specifically, we design a decentralized control protocol such that each agent achieves a predefined geometric formation with its initial neighbors, while using local information based on a limited sensing radius. The latter implies that the proposed scheme guarantees that the initially connected agents remain always connected. In addition, by introducing certain distance constraints, we guarantee inter-agent collision avoidance as well as collision avoidance with the obstacles and the boundary of the workspace. 
	Finally, simulation results verify the validity of the proposed framework.
\end{abstract}
\vspace{-3mm}
\section{Introduction}

During the last decades, decentralized control of multi-agent systems has gained a significant amount of attention due to the great variety of its applications, including  multi-robot systems, transportation, multi-point surveillance and biological systems. The main focus of multi-agent systems is the design of distributed control protocols in order to achieve global tasks, such as consensus, and at the same time fulfill certain properties, e.g., network connectivity. 

A particular multi-agent problem that has been considered in the literature is the formation control problem, where the agents represent robots that aim to form a prescribed geometrical shape, specified by a certain set of desired relative configurations. The main categories of formation control that have been studied in the related literature are (\cite{oh_park_ahn_2015}) position-based control, displacement-based control, distance-based control and orientation-based control. 

In position-based formation control, the agents control their positions to achieve the desired formation, as prescribed by some desired position offsets with respect to a global coordinate system. When orientation alignment is considered as a control design goal, the problem is known as orientation-based (or bearing-based) formation control. The desired formation is then defined by relative inter-agent orientations. The orientation-based control steers the agents to configurations that achieve desired relative orientation angles. In this work, we aim to design decentralized control protocols such that both position- and orientation-based formation are achieved.

The literature in position-based formation control is rich, and is traditionally categorized in single or double integrator agent dynamics and directed or undirected communication topologies (see e.g. \cite{beard_2001_coordination, egerstedt_formation, fax_murray_2004, do_2007_formation, dong_farrell_2008_cooperative, anderson_yu_fidan_hendrickx_2008, sepulchre_2008_symmeetric_formation, krick_broucke_francis_2009, dorfler_francis_2009, lin_jia_2010_rotating_formation, mesbahi_2010_graph_theory, cao_morse_yu_anderson_dagsputa_2011, belabbas2012robustness, zavlanos_2008_distributed}). Orientation-based formation control has been addressed in \cite{basiri_2010_angle_formation, eren_2012_bearing_formation, zhao2016bearing, oh_ahn_2014_angle_based_formation}, whereas the authors in \cite{oh_ahn_2014_angle_based_formation, bishop_2015_distributed, fathian2016globally} have considered the combination of position- and orientation-based formation.

The dominant case in the related literature of formation control is the two-dimensional one with simple dynamics and point-mass agents. In real applications, however, the engineering systems may have nonlinear $2$nd order dynamics, for which due to imperfect modeling the exact model is not a priori known. Other objectives concern connectivity maintenance, collision avoidance between the agents as well as collision avoidance between the agents and potential obstacles of the workspace, which renders the formation control problem a particularly challenging task. According to the authors' best knowledge, the combination of the aforementioned specifications has not been addressed in the related literature.

Motivated by this, we aim to address here the position-based formation control problem with orientation alignment for a team of rigid bodies operating in $3$D space, with $2$nd order nonlinear dynamics. We propose a decentralized control protocol that guarantees a geometric prescribed position- and orientation-based formation between initially connected agents. The proposed methodology
guarantees inter-agent collision avoidance and collision avoidance with the obstacles and the boundary of the workspace. In parallel, connectivity maintenance of the initially connected agents as well as representation singularity avoidance are ensured. In order to deal with the aforementioned specifications, we employ a novel class of potential functions. A special case of correct-by-construction potential functions, namely \textit{navigation functions}, has been introduced in \cite{koditschek1990robot} for the single-robot navigation, and has been employed in multi-agent formation control in \cite{jadbabaie_nf_formation, tanner_2005_formation_nf,kan2012network,dixon_2014_col_avoid_nf,dimarogonas2010analysis,tanner2012multiagent}. A more general potential function framework has been employed in \cite{do_2007_formation}. The aforementioned works, however, have only addressed the single integrator case, with no straightforward extension to higher-order systems. The authors in \cite{dimos2005_nf_2nd_order} deal with the double integrator case, but the goal was only navigation of the agents to specific points.  

In our previous work \cite{alex_chris_ppc_formation_ifac}, we treated a similar problem by utilizing a Prescribed Performance Control (PPC) scheme instead (for PPC controller design we refer to \cite{bechlioulis_tac_2008}), while only guaranteeing collision avoidance between neighboring agents forming a tree, with no obstacles or representation singularity avoidance. 
The main contribution of this paper is a novel decentralized control protocol scheme that generalizes \cite{alex_chris_ppc_formation_ifac} and solves a wider class of problems of multiple rigid bodies under Lagrangian dynamics with guaranteed collision avoidance among the agents, collision avoidance between agents and obstacles as well as singularity avoidance. 

The remainder of the paper is structured as follows. Section \ref{sec:preliminaries} gives the necessary notation and section \ref{sec:prob_formulation} provides the system and the problem statement. Section \ref{sec:solution} illustrates the proposed solution and Section \ref{sec:simulation_results} is devoted to a simulation example. Finally, conclusions and future work are discussed in Section \ref{sec:conclusions}.
\section{Notation} \label{sec:preliminaries}

The set of positive integers is denoted by $\mathbb{N}$. The real $n$-coordinate space, with $n\in\mathbb{N}$, is denoted by $\mathbb{R}^n$;
$\mathbb{R}^n_{\geq 0}$ and $\mathbb{R}^n_{> 0}$ are the sets of real $n$-vectors with all elements nonnegative and positive, respectively. Given a set $S$, we denote by $\lvert S\lvert$ its cardinality and by $S^N = S \times \dots \times S$ its $N$-fold Cartesian product. The notation $\|x\|$ is used for the Euclidean norm of a vector $x \in \mathbb{R}^n$. Define by $I_n \in \mathbb{R}^{n \times n}, 0_{m \times n} \in \mathbb{R}^{m \times n}$ the identity matrix and the $m \times n$ matrix with all entries zeros, respectively. A matrix $S \in \mathbb{R}^{n \times n}$ is called skew-symmetric if and only if $S^\top = -S$; $\mathcal{B}(c,r) = \{x \in \mathbb{R}^3: \|x-c\| \leq r\}$ is the $3$D sphere of center $c\in\mathbb{R}^{3}$ and radius $r \in \mathbb{R}_{> 0}$ and $\mathring{B}(c,r)$ its interior. Given a scalar function $y:\mathbb{R}^{n}\to\mathbb{R}$ and a vector $x\in\mathbb{R}^n$, denote by $\nabla_{x}y(x) = \tfrac{\partial y(x)}{\partial x} = [\tfrac{\partial y(x)}{\partial x_1},\dots, \tfrac{\partial y(x)}{\partial x_n}]^\top \in\mathbb{R}^n$ the gradient of $y$. The vector connecting the origins of coordinate frames $\{A\}$ and $\{B$\} expressed in frame $\{C\}$ coordinates in $3$D space is denoted by $p^{\scriptscriptstyle C}_{{\scriptscriptstyle B/A}}\in{\mathbb{R}}^{3}$. 
We further denote by $q_{\scriptscriptstyle B/A} = [\phi_{\scriptscriptstyle B/A}, \theta_{\scriptscriptstyle B/A}, \psi_{\scriptscriptstyle B/A}]^\tau \in\mathbb{T}$ the Euler angles representing the orientation of frame $\{B\}$ with respect to frame $\{A\}$,  with $\mathbb{T}=[-\pi,\pi]\times[-\tfrac{\pi}{2},\tfrac{\pi}{2}]\times[-\pi,\pi]$.
The angular velocity of frame $\{B\}$ with respect to $\{A\}$, expressed in frame $\{C\}$ coordinates, is denoted by $\omega^{\scriptscriptstyle C}_{\scriptscriptstyle B/A}\in \mathbb{R}^{3}$. We also use the notation $\mathbb{M} = \mathbb{R}^3\times \mathbb{T}$. For notational brevity, when a coordinate frame corresponds to an inertial frame of reference $\{0\}$, we will omit its explicit notation (e.g., $p_{\scriptscriptstyle B} = p^{\scriptscriptstyle 0}_{\scriptscriptstyle B/0}, \omega_{\scriptscriptstyle B} = \omega^{\scriptscriptstyle 0}_{\scriptscriptstyle B/0}$). All vector and matrix differentiations are derived with respect to the inertial frame $\{0\}$, unless otherwise stated.
\section{Problem Formulation} \label{sec:prob_formulation}
Consider a set of $N$ rigid bodies, with $\mathcal{V} = \{ 1,2, \ldots, N\}$, 

\noindent $N  \geq 2$, operating in a workspace $W \subseteq \mathbb{M}$, with coordinate frames $\{i\}, i\in\mathcal{V}$, attached to their centers of mass. The workspace is assumed to be modeled as a bounded sphere $W = \mathring{\mathcal{B}}(p_w,r_w)$ with center $p_w$ and radius $r_w$. Without loss of generality, we assume that $p_w = 0_{3 \times 1}$, representing an inertial reference frame $\{0\}$. The subscript $w$ stands for the workspace $W$. We consider that each agent occupies a sphere $\mathcal{B}(p_i, r_i)$, where $p_i\in \mathbb{R}^3$ is the position of the agent's center of mass and $r_i < r_w$ is the agent's radius. We also denote by $q_i\in\mathbb{T}, i\in\mathcal{V}$, the Euler angles representing the agents' orientation with respect to $\{0\}$, with $q_i = [\phi_i,\theta_i,\psi_i]^\top$. By defining $x_i\in W, v_i \in\mathbb{R}^6,$ with $x_i = [p^\top_i,q^\top_i]^\top, v_i =[\dot{p}^\top_i,\omega^\top_i]^\top$, we model each agent's motion with the $2$nd order dynamics: 
\vspace{-1mm}
\begin{subequations}\label{eq:system} 
	\begin{align} 
		&\dot{x}_i = J^{-1}_i(q_i)v_i , \label{eq:system_1} \\ 
		& M_i(x_i) \dot{v}_i + C_i(x_i,\dot{x}_i) v_i+g_i(x_i)= u_i,  \label{eq:system_2} 
	\end{align}
\end{subequations}
where $J_i:\mathbb{T} \to \mathbb{R}^{6\times6}$ is a \textit{Jacobian matrix} that maps the Euler angle rates to $v_i$, given by 
\begin{align} 
	J_i(q_i) 
	&=
	\begin{bmatrix}
		I_3 & 0_{3 \times 3} \\
		0_{3 \times 3} & J_{q_i}(q_i) \\
	\end{bmatrix} \notag, \\
	J_{q_i}(q_i) 
	&= 
	\begin{bmatrix}
		1 & 0  & \sin(\theta_i) \\
		0 & \cos(\phi_i) & -\cos(\theta_i) \sin(\phi_i) \\
		0  & \sin(\phi_i) & \cos(\phi_i) \cos(\theta_i)
	\end{bmatrix}, \notag
\end{align}
and $J^{-1}_i(q_i)$ is its matrix inverse.
The matrix $J_i$ is singular when $\det(J_i) = \cos(\theta_i) = 0\Leftrightarrow \theta_i = \pm\tfrac{\pi}{2}$, which we refer to as \textit{representation singularity}. The proposed controller will guarantee, however, that this is always avoided and thus \eqref{eq:system_1} is well defined.  

Furthermore, $M_i: W \to \mathbb{R}^{6\times6}$ is the positive definite \emph{inertia matrix}, $C_i: W \times \mathbb{R}^6 \to \mathbb{R}^{6\times6}$ is the \emph{Coriolis matrix}, and $g_i: W \to \mathbb{R}^6$ is the \emph{gravity vector}. We consider that the Coriolis and the inertia vector fields are \emph{unknown}.
Finally, $u_i \in \mathbb{R}^6$ is the control input vector representing the $6$D generalized \emph{actuation force} acting on agent $i\in\mathcal{V}$. Let us also define the stack vectors $x = [x_1^\top,\dots,x_N^\top]^\top \in W^N$ and $v = [v_1^\top, \dots, v_N^\top]^\top \in \mathbb{R}^{6N}$. 
In addition, the matrices $\dot{M}_i - 2C_i$ are skew-symmetric \cite{Siciliano2009}, i.e., $y^\top \left[\dot{M}_i - 2 C_i\right]y = 0, \forall y \in \mathbb{R}^6, i \in \mathcal{V}$.
For the state measurement of each agent, the following assumption is required.

\begin{assumption} (Measurements Assumption)
	Each agent $i$ can measure its own states $p_i,q_i, \dot{p}_i, v_i, i\in\mathcal{V}$, and has a limited sensing range of: $d_i > \max\{r_i+r_j: i,j \in \mathcal{V}, i \neq j\}.$ 
\end{assumption}
Therefore, by defining the neighboring set as $\mathcal{N}_i(x_i) = \{j\in\mathcal{V} : p_j\in\mathcal{B}(p_i, d_i),i\neq j\}$, in view of the aforementioned assumption, agent $i$ knows at each configuration $x_i$ all $p^i_{j/i}$, $q_{j/i}$ and, since it knows its own $p_i,q_i$, it can compute all $p_{j}$, $q_{j}$, $\forall \ j \in \mathcal{N}_i(x_i), x_i\in W$. For the neighboring set $\mathcal{N}_i(x_i)$ define also $N_i(x_i) = |\mathcal{N}_i(x_i)|$. 

Moreover, we consider that in the given workspace there exist $Z\in\mathbb{N}$ \emph{static obstacles}, with $\mathcal{Z} \triangleq \{1,\dots,Z\}$, modeled as the spheres $\mathcal{B}(p_{o_z}, r_{o_z})$, with centers and radii $p_{o_z}\in\mathbb{R}^3, r_{o_z}\in \mathbb{R}_{>0}, z \in \mathcal{Z}$, respectively. The geometry of two agents $i,j$ and an obstacle $z$ in the workspace $W$ is depicted in Fig. \ref{fig:agents_geometry}.

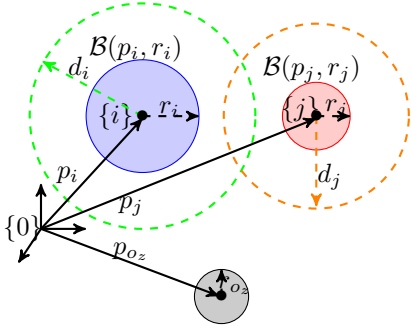
\begin{figure}[t!]
\vspace{0.3cm}
	\centering
	\begin{tikzpicture}[scale = 0.3]	
	\draw [color=black,thick,->,>=stealth'](-9, -5) to (-7, -5);
	\draw [color=black,thick,->,>=stealth'](-9, -5) to (-9, -3);
	\draw [color=black,thick,->,>=stealth'](-9, -5) to (-10, -6.5);
	\node at (-9.8, -5.0) {$\{0\}$};
	
	\draw [color = blue, fill = blue!20] (-4.5,0) circle (2.5cm);
	\node at (-5.7, 0.0) {$\{i\}$};
	\draw[green,thick,dashed] (-4.5,0) circle (5.0cm);
	\draw [color=black,thick,->,>=stealth'](-9, -5) to (-4.5, -0.1);
	\node at (-7.80, -2.6) {$p_i$};
	\draw [color=green,thick,dashed,->,>=stealth'](-4.5, 0.0) to (-8.93, 2.43);
	\node at (-7.3, 2.15) {$d_i$};
	\draw [color=black,thick,dashed,->,>=stealth'](-4.5, 0.0) to (-2.0, 0.0);
	\node at (-3.3, 0.3) {$r_i$};
	\node at (-4.5, 0.0) {$\bullet$};
	\node at (-4.8, 3.0) {$\mathcal{B}(p_i, r_i)$};
	
	\draw [color = red, fill = red!20] (3.2, 0) circle (1.5cm);
	\node at (2.5, 0.3) {$\{j\}$};
	\draw[orange,thick,dashed,] (3.2, 0) circle (4.1cm);
	\draw [color=black,thick,->,>=stealth'](-9, -5) to (3.2, -0.1);
	\node at (-5.0, -4.0) {$p_j$};
	\draw [color=orange,thick,dashed,->,>=stealth'](3.2, 0.0) to (3.2, -4.0);
	\node at (3.8, -2.7) {$d_j$};
	\draw [color=black,thick,dashed,->,>=stealth'](3.2, 0.0) to (4.7, 0.0);
	\node at (4.1, 0.3) {$r_j$};
	\node at (3.2, 0.0) {$\bullet$};
	\node at (3.0, 2.1) {$\mathcal{B}(p_j, r_j)$};
	
	\draw [color = black, fill = black!20] (-1, -8) circle (1.2cm);
	\draw [color=black,thick,->,>=stealth'](-9, -5) to (-1.1, -7.98);
	\draw [color=black,thick,dashed,->,>=stealth'](-1, -8) to (-1, -6.8);
	\node at (-1, -8) {$\bullet$};
	\node at (-5.0, -6.0) {$p_{o_z}$};
	\node at (-0.40, -7.5) {$r_{o_z}$};
	\end{tikzpicture}
	\caption{Illustration of two moving agents $i, j \in \mathcal{V}$ and a static obstacle $o_z$ in the workspace; $\{0\}$ is the inertial frame, $\{i\}, \{j\}$ are the frames attached to the agents' center of mass, $p_i, p_j, p_{o_z} \in \mathbb{R}^3$ are the positions of the center of mass of the agents $i,j$ and the obstacle $o_z$, respectively, with respect to $\{0\}$; $r_i, r_j, r_{o_z}$ are the radii of the agents $i,j$ and the obstacle $o_z$, respectively; $d_i, d_j$ with $d_i > d_j$ are the agents' sensing ranges. Note that the agents are not neighbors since $p_j \notin \mathcal{B}(p_i, d_i)$ and $p_i \notin \mathcal{B}(p_j, d_j)$.}
	\label{fig:agents_geometry}
\end{figure}
Let us define the distances $d_{ij,a}: \mathbb{R}^6 \to \mathbb{R}_{\geq 0}$, $d_{iz,o}:\mathbb{R}^3 \to \mathbb{R}_{\geq 0}$, with: $d_{ij,a}(p_i,p_j) = \| p_i - p_j \|, d_{iz,o}(p_i) = \| p_i - p_{o_z} \|$, $\forall i,j\in\mathcal{V}, i\neq j,z\in\mathcal{Z}$, as well as the constants $\underline{d}_{ij, a} = r_{i} + r_{j}, \underline{d}_{iz, o} = r_{i} + r_{o_z}$, that represent the minimum distance such that agents $i$ and $j$ and agent $i$ and object $z$, do not collide, respectively. The subscripts $a$ and $o$ stand for \textit{agent} and \textit{obstacle}, respectively. The following assumption is required, for the feasibility of the problem:
\begin{assumption} \label{as:geometry}
It holds that 
\begin{enumerate}
\item $\lVert p_{o_z} - p_{o_{z'}} \rVert \geq 2\max\limits_{i\in\mathcal{V}}\{r_i\} + r_{o_z} + r_{o_{z'}} + \varepsilon_r, \forall z,z'\in\mathcal{Z}$, with $z\neq z'$,
\item $r_w - (\lVert p_{o_z} \rVert + r_{o_z}) \geq 2\max\limits_{i\in\mathcal{V}}\{r_i\} + \varepsilon_r, \forall z\in\mathcal{Z}$, 
\end{enumerate}
where $\varepsilon_r$ is an arbitrarily small positive constant.
\end{assumption}
The aforementioned assumption states that there is enough space between the obstacles and the workspace boundary as well as the obstacles themselves for the agents to navigate among them. 

Due to the fact that the agents are not dimensionless and their sensing capabilities are limited, the control protocol, except from achieving desired position formation (define it by $p_{ij, \text{des}}$) and desired formation angles (define it by $q_{ij, \text{des}}$) for all neighboring agents $i \in \mathcal{V}, j \in \mathcal{N}_i(x_i(0))$, it should also guarantee for all $t\in\mathbb{R}_{\geq 0}$ that (i) all the agents avoid collision with every other agent; (ii) all the agents avoid collision with all the obstacles; (iii) all the agents avoid collision with the workspace boundary, (iv) all the initial edges are maintained, i.e., connectivity maintenance, and (v) the singularity of the Jacobian  matrices $J_i$ is avoided. 



\begin{definition} \label{def:feasible_formation}
	(Feasible Formation) Given the initial neighboring sets $\mathcal{N}_i(x_i(0)),i\in\mathcal{V}$, the desired displacements $x_{ij,\text{des}} = [p^\top_{ij,\text{des}},q^\top_{ij,\text{des}}]^\top$ that characterize a formation configuration, are called \textit{feasible} if $|\theta_{ij,\text{des}}| <\pi$, $\forall j\in\mathcal{N}_i(x_i(0)),i\in\mathcal{V}$,  and 
	$\bigcap_{i \in V} \{x_i \in W  : \lVert x_i - x_j -x_{ij,\text{des}} \rVert = 0, d_{iz,o}(p_i) > 0, \lVert p_i \rVert + r_i < r_w, \forall z\in\mathcal{Z},j\in\mathcal{N}_i(x_i(0)) \} \neq \emptyset$.
\end{definition}



Formally, the control problem under the aforementioned constraints is formulated as follows:
\begin{problem} \label{problem}
	Consider $N$ agents governed by the dynamics \eqref{eq:system} and operating in a workspace $W$ with $Z$ spherical obstacles, with:
	\begin{itemize}
		\item $v_i(0) = 0_{6 \times 1}, \forall i\in\mathcal{V}$,
		\item $-\frac{\pi}{2} < -\bar{\theta} \leq \theta_i(0) \leq \bar{\theta} < \tfrac{\pi}{2}, \forall i\in\mathcal{V}$,
		\item $\lVert p_i(0) - p_j(0) \rVert > \underline{d}_{ij,a}, \forall i,j\in\mathcal{V}, i\neq j$,
		\item $\lVert p_i(0) - p_{o_z}(0) \rVert > \underline{d}_{iz,o}, \forall i\in\mathcal{V},z\in\mathcal{Z}$,
	\end{itemize}
	i.e., singularity- and collision- free configurations at $t=0$, where $\bar{\theta}$ is an arbitrary constant in the open set $(0,\frac{\pi}{2})$.
	Then, given a nonempty initial set $\mathcal{N}_i(x_i(0)) \neq \emptyset$, \textit{feasible} inter-agent displacements $p_{ij, \text{des}}, q_{ij, \text{des}}, \forall i \in \mathcal{V}, j \in \mathcal{N}_i(x_i(0))$, such that $\underline{d}_{ij,a} < d_{ij,a}(p_i,p_j) < d_{i}, \forall (p_i,p_j)\in\{(p_i,p_j)\in \mathbb{R}^6 : \lVert x_i - x_j - x_{ij,\text{des}} \rVert = 0 \}$, design decentralized control laws $u_i$, such that for every $i \in \mathcal{V}$ the following hold:  
	\begin{enumerate}
		\item 
		
		$\lim\limits_{t \to \infty} \lVert x_{i}(t)-x_{j}(t) - x_{ij, \text{des}} \rVert =0, \forall j\in\mathcal{N}_i(x_i(0)),$\\
		
		
		
		\item 
		$\|p_i(t)-p_j(t)\| > \underline{d}_{ij, a}, \forall j\in\mathcal{V}\backslash\{i\},  t \in \mathbb{R}_{\geq 0}$,\\
		
		\item 
		$\lVert p_i(t) - p_{o_z}(t) \rVert > \underline{d}_{iz,o}, \forall z\in\mathcal{Z}, t \in \mathbb{R}_{\geq 0}$,\\
		
		\item $\|p_i(t)\|+r_i < r_w, \forall t \in \mathbb{R}_{\ge 0}$. \\	
		
		\item 
		$\|p_i(t)-p_j(t)\| < d_{i}, \forall j \in \mathcal{N}_i(x_i(0)), t \in \mathbb{R}_{\geq 0}$,\\
		
		\item $	-\frac{\pi}{2} < -\bar{\theta}  \leq \theta_i(t) \leq \bar{\theta} < \frac{\pi}{2}$, $\forall t\in\mathbb{R}_{\ge 0}$,	
	\end{enumerate}
\end{problem}
The aforementioned specifications imply the following: $1)$ stands for formation control (both position and orientation); $2)$ stands for inter-agent collision avoidance; $3)$ stands for collision avoidance between the agents and the obstacles; $4)$ stands for collision avoidance between the agents and the boundary; $5)$ stands for connectivity maintenance of the initially connected agents and finally, $6)$ stands for the representation of singularities avoidance.

\section{Problem Solution} \label{sec:solution}

In this section, a systematic solution to Problem \ref{problem} is introduced. 
In particular, the following analysis is performed: First, the form of the proposed potential function along with its components is described. Then, we provide the proposed decentralized controllers that guarantee the satisfaction of all the control specifications. The required stability analysis is presented subsequently.

In order to solve the formation control problem with the collision- and singularity- avoidance as well as connectivity maintenance, we use a \emph{decentralized potential function} for each agent $i \in \mathcal{V}$ as $\varphi_i(x)$, with the following properties:
\begin{enumerate}[(i)]
	\item The function $\varphi_i(x)$, is not defined, i.e., $\varphi_i(x) = \infty$, $\forall i\in\mathcal{V}$, when a collision or a connectivity break occurs,
	\item The points where $\nabla_{x_i}\varphi_i(x) = 0$, $\forall i\in\mathcal{V}$, consist of the goal configurations and a set of configurations whose region of attraction (by following the vector field curves) is a set of measure zero.	
	\item It holds that $\nabla_{x_i} \left(\varphi_i(x) + \sum\limits_{j\in\mathcal{N}_i(x_i)}\varphi_j(x)\right)=0$ $\Leftrightarrow \nabla_{x_i}\varphi_i(x) = 0$ and $\nabla_{x_i}\sum\limits_{j\in\mathcal{N}_i(x_i)} \varphi_j(x) = 0$, $\forall i\in\mathcal{V}$, $x\in W^N$.
\end{enumerate}
More specifically, $\varphi_i(x)$ is a function of two main terms, a  \emph{goal function} $\gamma_i(x)$, that should vanish at the desired configuration, and an \emph{obstacle function}, 
$\beta_i(x)$ is a function that encodes inter-agent collisions, collisions between the agents and the obstacle boundary/undesired regions of interest, connectivity losses between initially connected agents and singularities of the Jacobian matrix $J_i(q_i)$; $\beta_i(x)$ vanishes when one or more of the above situation occurs.
Next, we provide a construction of the goal and obstacle terms. However, the construction of $\varphi_i$ is out of the scope of this paper. Examples can be found in \cite{panagou2017distributed,dimarogonas2007decentralized,dimarogonas2006feedback}\footnote{For instance, we can choose $\varphi_i = \tfrac{1}{1-\phi_i}$, where $\phi_i$ is the navigation function used in \cite{dimarogonas2007decentralized,dimarogonas2006feedback}.}. 



\subsubsection{$\gamma_i(x)$ - Goal Function}

The function $\gamma_i : W^{N} \to \mathbb{R}_{\ge 0}$ encodes the control objective of agent $i$, which is to achieve position and orientation formation with its neighboring agents. With that in mind, a reasonable choice of $\gamma_i(x)$ is:
\vspace{-3mm}
\begin{align} \label{eq:gamma_i}
	\gamma_i(x) &= \sum_{j \in \mathcal{N}_i(0)} \left\{ \gamma_{ij, p}(p_i, p_j) + \gamma_{ij, q}(q_i, q_j) \right\}, \notag \\ 
	&= \sum_{j \in \mathcal{N}_i(0)}  \gamma_{ij, x}(x_i, x_j) , 
\end{align}
where $\gamma_{ij, p}(p_i, p_j) =  \| p_i - p_j - p_{ij,\text{des}}\|^2, \gamma_{ij, q}(q_i, q_j) =  \| q_i - q_j - q_{ij,\text{des}}\|^2, \gamma_{ij, x}(x_i, x_j) = \| x_i - x_j - x_{ij,\text{des}}\|^2$. 
The function $\gamma_i(x)$ reaches its unique global minimum when both $p_i - p_j = d_{ij,\text{des}}$ and $q_i - q_j = q_{ij,\text{des}}, \forall i\in\mathcal{V}$, i.e., when both formation and orientation alignment are achieved between all the neighboring agents.

\subsubsection{$\beta_i(x)$ - Obstacle Function} \label{sec:beta_gamma_definitions}

The function $\beta_i(x) : W^{N} \to \mathbb{R}$, encodes all inter-agent collisions, collisions between the agents and obstacles, collisions with the boundary of the workspace, connectivity between initially connected agents and singularities of the Jacobian matrix $J_i(x_i)$. First, for each agent $i$, we define the functions $\eta_{ij,a}:\mathbb{R} \to \mathbb{R}_{\geq 0}, \eta_{iz,o}:\mathbb{R} \to \mathbb{R}_{\geq 0}, \eta_{ij,c}:\mathbb{R} \to \mathbb{R}_{\geq 0}$ where:
	\begin{align*}
		\eta_{ij,a}(d_{ij,a}) &= d^2_{ij,a}-\underline{d}^2_{ij,a}, \\
		\eta_{iz,o}(d_{iz,o}) &= d^2_{iz,o}-\underline{d}^2_{iz,o}, \\
		\eta_{ij,c}(d_{ij,a}) &= d^2_i-d^2_{ij,a}. 
	\end{align*}
The subscripts $j$ and $z$ correspond to agent $j\in\mathcal{V}\backslash\{i\}$ and obstacle $z\in\mathcal{Z}$, respectively, whereas the subscript $c$ stands for \textit{connectivity}. Let us also define the functions $b_{ij,a}: \mathbb{R}_{\geq 0} \to [0,1]$, $b_{iz,o}: \mathbb{R}_{\geq 0} \to [0,1]$, $b_{ij,c}: \mathbb{R}_{\geq 0} \to [0,1]$, $b_{iw}:\mathbb{R}_{\geq 0} \to \mathbb{R}$, $b_{J_i}: [-\tfrac{\pi}{2},\tfrac{\pi}{2}] \to [0,1]$, where:
	\begin{align*}
		&\hspace{-2mm} b_{ij,a}(x) =
		\begin{cases}
			\phi_{i,a}(x), &  0 \le x < d^2_i - \underline{d}^2_{ij,a} , \\
			1, &  d^2_i - \underline{d}^2_{ij,a} \le x, \\
		\end{cases} \\ 
		&\hspace{-2mm} b_{iz,o}(x) = 
		\begin{cases}
			\phi_{i,o}(x), &  0 \le x < d^2_i - \underline{d}^2_{iz,o}, \\
			1, & d^2_i - \underline{d}^2_{iz,o} \le x, 
		\end{cases} \\ 
		&\hspace{-2mm}b_{ij,c}(x) = 
		\begin{cases}
			0 &  x < 0,\\
			\phi_{i,c}(x), &   0 \le x < d^2_i - \underline{d}^2_{ij,a}, \\
			1, &  d^2_i - \underline{d}^2_{ij,a} \le x, 
		\end{cases} \\ 
		&\hspace{-2mm} b_{iw}(x) = \left[1 - \frac{x}{(r_w-r_i)^2}\right]^2, \\ 
		&\hspace{-2mm} b_{J_i}(x) = \cos^2(x). 
	\end{align*}
The functions $\phi_{i,a}, \phi_{i,o}, \phi_{i,c}$ are \emph{increasing polynomials}, appropriately selected to guarantee that the functions $b_{ij,a}, b_{iz,o}, b_{ij,c}$, respectively, are twice continuously differentiable everywhere, with $\phi_{i,a}(0)$ $= \phi_{i,o}(0)$ $= \phi_{i,c}(0) = 0$, $\forall i\in\mathcal{V}$. 
The functions $b_{iz,o}, b_{ij,c}$ have an identical behavior. We can now choose the function $\beta_i: W^{N} \to [0,1]$ as the following product for every $i \in \mathcal{V}$:
\vspace{-2mm}
\begin{align}
	& \beta_i(x) = b_{iw}(\lVert p_i \rVert^2) b_{J_i}(\theta_i) \left[\prod_{j\in\mathcal{V}\backslash\{i\}}^{} b_{ij,a}(\eta_{ij,a})\right] \notag\\
	&\hspace{15mm} \Bigg[\prod_{z\in\mathcal{Z}}^{} b_{iz,o}^{}(\eta_{iz,o}) \Bigg] \left[\prod_{j\in\mathcal{N}_i(0)}^{} b_{ij,c}^{}(\eta_{ij,c}) \right]. \label{eq:beta_function}
\end{align}
The functions $b_{ij,a}(\eta_{ij,a}), b_{iz,o}(\eta_{iz,o}), b_{ij,c}(\eta_{ij,c})$ correspond to inter-agent collision, collision with obstacles and connectivity maintenance, respectively, for agent $i\in\mathcal{V}$, while the functions $b_{iw}(\lVert p_i \rVert^2), b_{J_i}(\theta_i)$ correspond to collision with the workspace boundary and representation singularities. Each of these terms becomes zero when there is a collision, a connectivity break or a representation singularity. 
Note that all the aforementioned functions use only local information depending on the sensing range $d_i$ of agent $i$. 

\begin{remark}
	Note that the choice of the functions $\gamma_i$ and $\beta_i$ from \eqref{eq:gamma_i} and \eqref{eq:beta_function}, respectively, renders $\varphi_i$ from \eqref{eq:gamma_i}, \eqref{eq:beta_function} to be a function only of the neighboring states of $x_i$  i.e., decentralized. The function $\gamma_i$ depends on the initial neighboring set $\mathcal{N}_i(x_i(0))$ since the formation requirements need to be achieved between the agents that belong to the initial neighboring set $\mathcal{N}_i(x_i(0))$. Furthermore, the function $\beta_i$ depends on the time-varying neighboring set $\mathcal{N}_i(x_i)$, since in order to capture the collision avoidance goals, the neighboring sets $\mathcal{N}_i(x_i)$ need to be updated with new potential neighbors.
\end{remark}


With the introduced notation, the properties of the functions $\varphi_i$ are: 
\begin{enumerate}[(i)]
	\item $\beta_i(x)\to 0 \Leftrightarrow \varphi_i(x) \to \infty, \forall i\in\mathcal{V}$,
	\item $ \nabla_{x_i}\varphi_i(x)|_{x_i=x^\star_i} = 0,  \forall x^\star_i\in W \text{ s.t. } \gamma_i(x^\star_i) = 0$ and the regions of attraction of the points $\{x \in W^{N}: \nabla_{x_i}\varphi_i(x)|_{x_i=\widetilde{x}_i} = 0, \gamma_i(\widetilde{x}_i) \neq 0 \}, i\in\mathcal{V}$, are sets of measure zero.
\end{enumerate}

Next, we design bounded controllers $u_i$ such that all the specifications of Problem \ref{problem} are met, according to the following theorem, which summarizes the main results of this work.
\begin{theorem}
The decentralized control law $u_i: \mathbb{D}_i  \times \mathbb{R}^{6} \to \mathbb{R}^6$, with $\mathbb{D}_i = \{x\in W^N: \beta_i(x) > 0\}$ and 
	\begin{align}
		u_i(x, v_i) =& -\widetilde{k}_{i} v_i + g_i(x_i) \notag \\
		&\hspace{-15mm} - \left[ J^{-1}_i(x_i) \right]^\top \left\{\nabla_{x_i} \left(\varphi_i(x) + \sum\limits_{j\in\mathcal{N}_i(x_i)}\varphi_j(x)\right) \right\} \label{eq:controller_design},
	\end{align}
	for each agent $i \in \mathcal{V}$, with control gain $\widetilde{k}_i> 0$, 
	brings agent $i$ to its desired configuration from almost all initial conditions, while ensuring $\beta_i > 0,\forall i\in\mathcal{V}$ and the boundedness of all closed loop signals, providing a solution to Problem \ref{problem}.
\end{theorem}
\begin{proof}
	Consider the nonnegative Lyapunov-like function for the system \eqref{eq:system} $L:\mathbb{D}_1\times\dots\times\mathbb{D}_N\times\mathbb{R}^{6N}$, with:
	\begin{align}
    \hspace{-3mm}L(x,v) = \sum_{i \in \mathcal{V}} \left\{ \varphi_i(x) + \frac{1}{2}v_i^\top M_i(x_i) v_i\right\}, \label{eq:L}
	\end{align}
	Since the system configuration at $t=0$ is singularity- and collision-free, the functions $\beta_i$ are strictly positive at $t=0, \forall i\in\mathcal{V}$. Furthermore, $v_i(0) = 0_{6 \times 1}, \forall i\in\mathcal{V}$. Thus, $L$ is initially bounded, i.e., there exists a positive and finite constant $M$ such that
	\begin{equation} \label{eq:V_0_bounded}
		L_0 \triangleq L(x(0), v(0)) \leq M.
	\end{equation}
	By differentiating \eqref{eq:L} with respect to time, substituting the dynamics \eqref{eq:system}, using the skew-symmetry of $\dot{M}_i - 2C_i, \forall i\in\mathcal{V}$ and using the fact that  $\sum_{i\in\mathcal{V}} ( [\nabla_{x_i}\varphi_i(x)]^\top\dot{x}_i + \sum_{j\in\mathcal{N}_i(x_i)} [\nabla_{x_j}\varphi_i(x)]^\top\dot{x}_j ) = \sum_{i\in\mathcal{V}}( [\nabla_{x_i}\varphi_i(x) ]^\top +  \sum_{j\in\mathcal{N}_i(x_i)} [\nabla_{x_i}\varphi_j(x)]^\top) \dot{x}_i$, we obtain:
	\begin{align}
		\dot{L}&= \sum_{i \in \mathcal{V}} v^\top_i\Big[ \left( J^{-1}_i(x_i) \right)^\top \nabla_{x_i}\varphi_i(x) + u_i -g_i(x_i)  \notag \\
		&\hspace{10mm} + \sum_{j\in\mathcal{N}_i(x_i)} [\nabla_{x_i}\varphi_j(x)] \Big] , \notag
	\end{align}
	which, by substituting the control law \eqref{eq:controller_design},   becomes $\dot{L} \le - \sum_{i \in \mathcal{V}} \widetilde{k}_{i} \|v_i\|^2$.
    Therefore,  $L$ is non-increasing and hence, in view of \eqref{eq:V_0_bounded}, we conclude that $L(x(t),v(t)) \le L_0 \le M, \forall t \in \mathbb{R}_{\ge 0}$,
	and the boundedness of all the terms $x(t)$, $v(t)$, $\varphi_i(x(t)), \forall t \in \mathbb{R}_{\ge 0}$. By invoking the properties of $\varphi_i(x)$, we also conclude that $\beta_i(x(t)) > 0$, $\forall t\in\mathbb{R}_{\geq 0}$ and hence,  
	inter-agent collisions, collisions between the agents and the obstacles/workspace boundary as well as connectivity losses and singularities, are avoided. 
	
	Moreover, by invoking LaSalle's Invariance Principle, the state of the system converges to the largest invariant set contained in the set $S = \left\{x,v: \dot{L}(x,v) = 0\right\} = \left\{x,v : v_i = 0_{6\times 1}, \forall i \in \mathcal{V}\right\}$.
	For $S$ to be invariant we require $\dot{v}_i = 0_{6\times1}$, from which, the closed-loop system implies $\left[J^{-1}_i(x_i)\right]^\top \nabla_{x_i}\varphi_i(x) = 0_{6\times1}$, where the properties of $\varphi_i(x)$ were exploited. Note that, since $\beta_i(x(t)) > 0$, $\forall t\in\mathbb{R}_{\geq 0}$, $J(x_i(t))$ is always nonsingular. Therefore, we conclude that the closed loop system will converge to the configuration where $\nabla_{x_i}\varphi_i(x) = 0_{6\times1}, \forall i\in\mathcal{V}$. According to the inherent properties of the functions $\varphi_i(x)$, this will happen from all initial conditions except for a set of measure zero, i.e., almost all initial conditions \cite{dimarogonas2006feedback,koditschek1992robot}. Moreover, it can be shown that
	the control laws \eqref{eq:controller_design} stay bounded $\forall t\in\mathbb{R}_{\geq 0}, i\in\mathcal{V}$.
	\begin{figure}[t!]
	   \vspace{0.3cm}
		\centering
		\includegraphics[scale = 0.25]{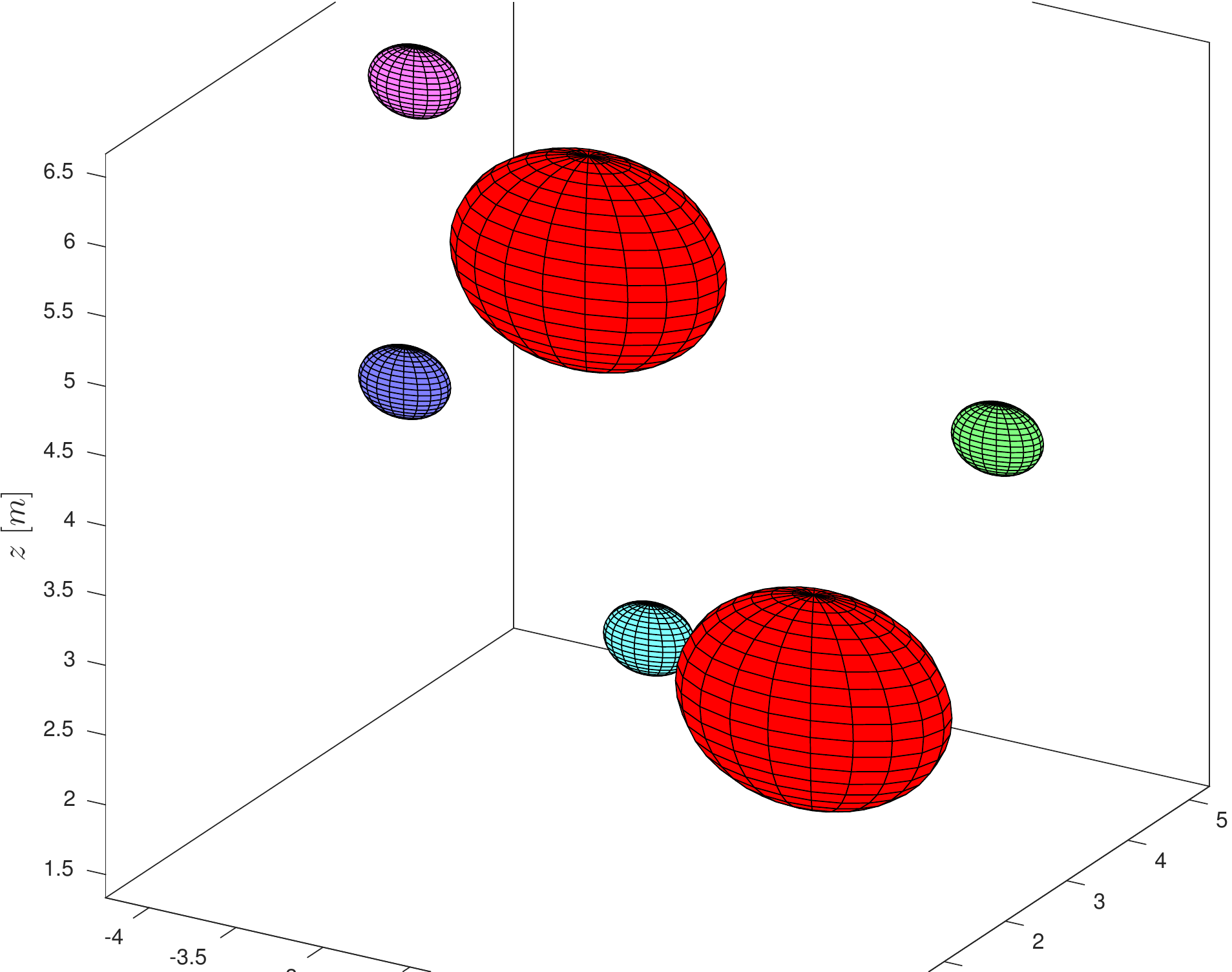}
		\caption{The initial workspace of the simulated scenario ($t=0$). Agent 1 (with blue), agent 2 (with  green), agent 3 (with cyan) and agent 4 (with purple) and two obstacles (with red).}
		\label{fig:init_ws}       
	\end{figure}	
	\begin{figure}[t!]
		\centering
		\includegraphics[trim={0 0 0cm -2cm},scale = 0.34]{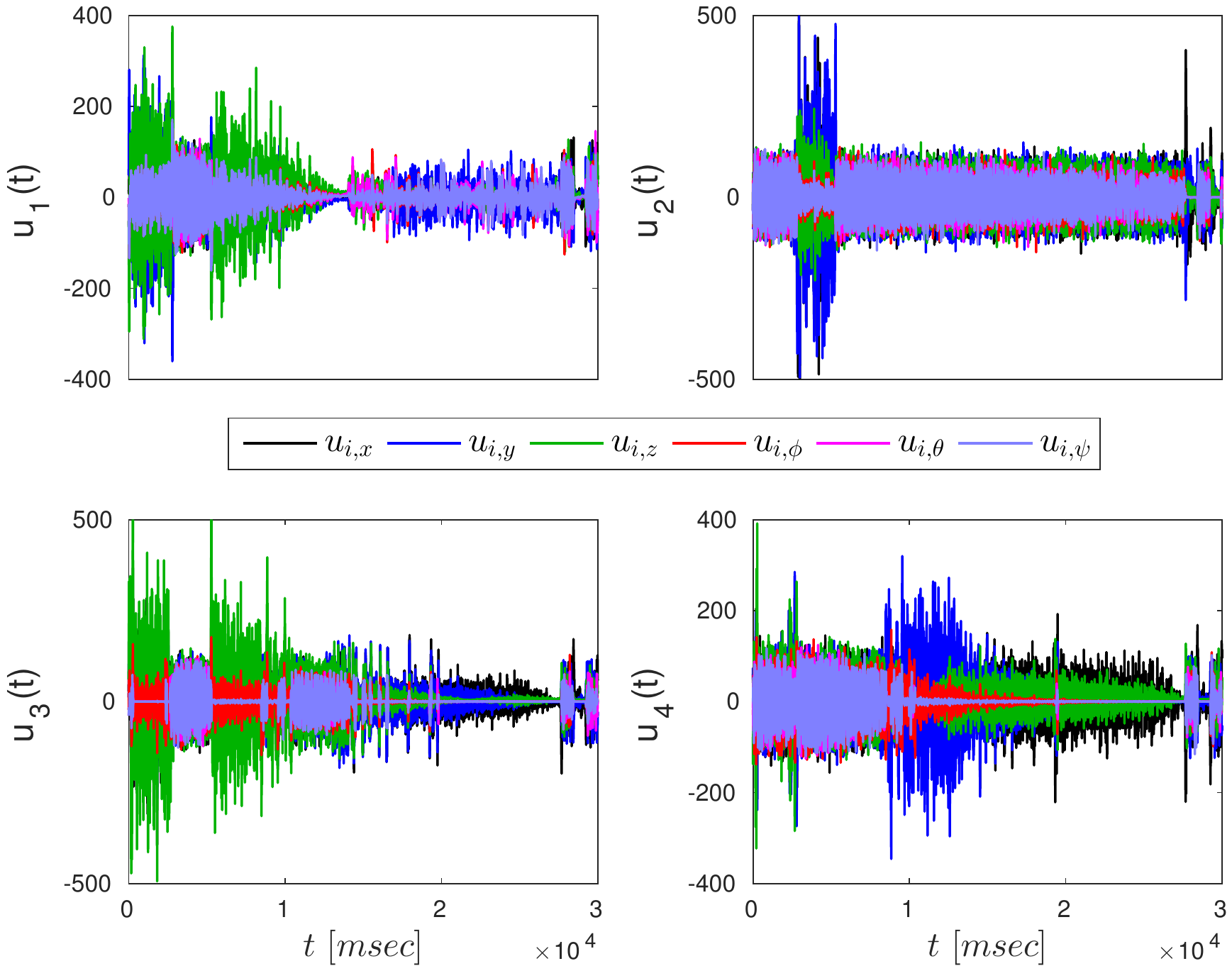}
		\caption{The resulting control inputs $u_i,\forall i\in\{1,\cdots,4\}$.}
		\label{fig:inputs}       
	\end{figure}
\end{proof}
\begin{remark}
	Note that the design of the obstacle functions \eqref{eq:beta_function} renders the control laws \eqref{eq:controller_design} decentralized, in the sense that each agent uses only local information with respect to its neighboring agents, according to its limited sensing radius. Each agent can obtain the necessary information to cancel the term $\sum_{j\in\mathcal{N}_i(x_i)} \nabla_{x_i}\varphi_j(x)$ from its neighboring agents. 
\end{remark}

\vspace{-2mm}

\begin{figure}[t!]
\vspace{0.3cm}
	\centering
	\includegraphics[scale = 0.35]{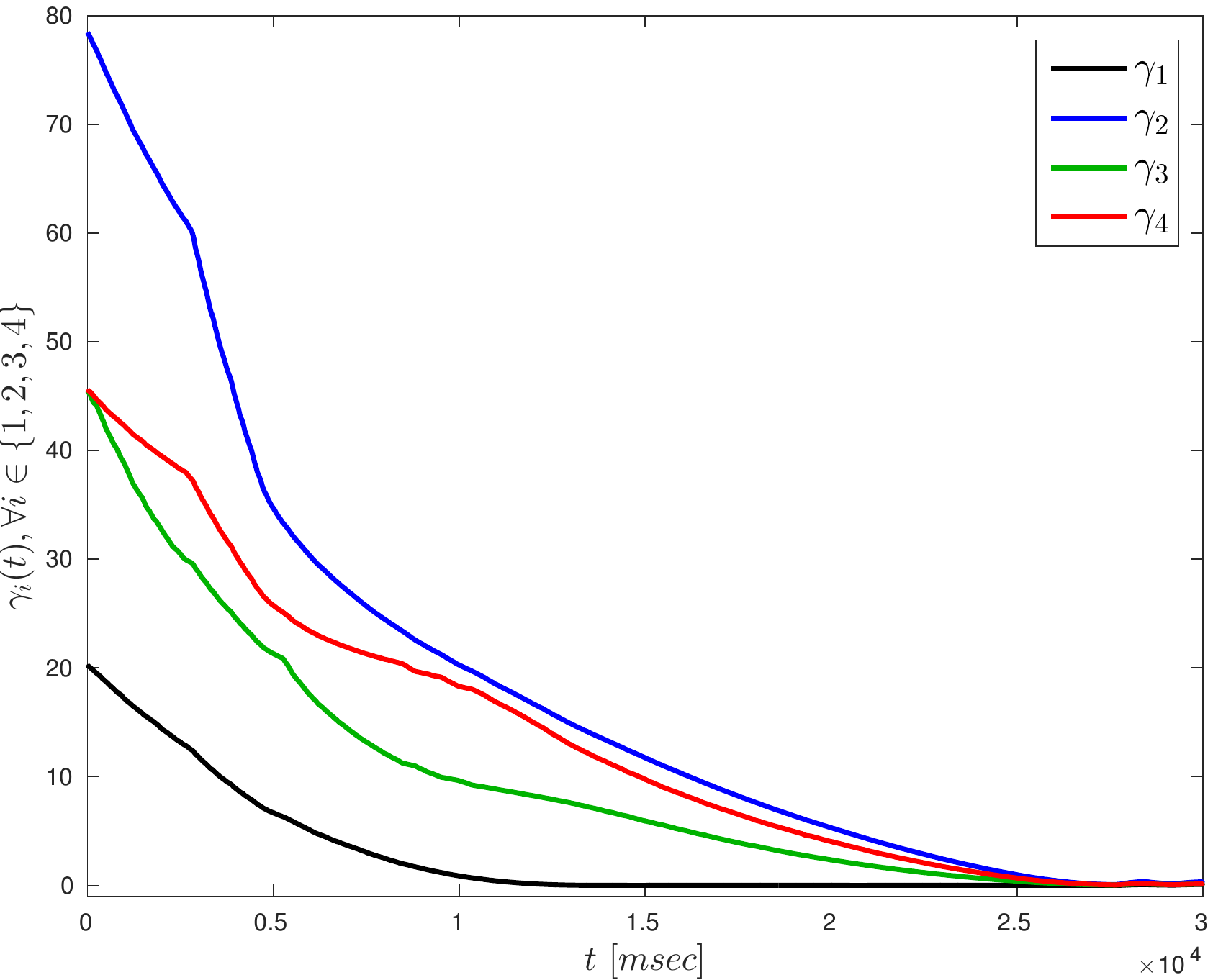}
	\caption{The evolution of the goal functions $\gamma_i,\forall i\in\{1,\cdots,4\}$, which are shown to converge to zero.}
	\label{fig:gammas}       
\end{figure}

\begin{figure}[t!]
	\centering
	\includegraphics[scale = 0.38]{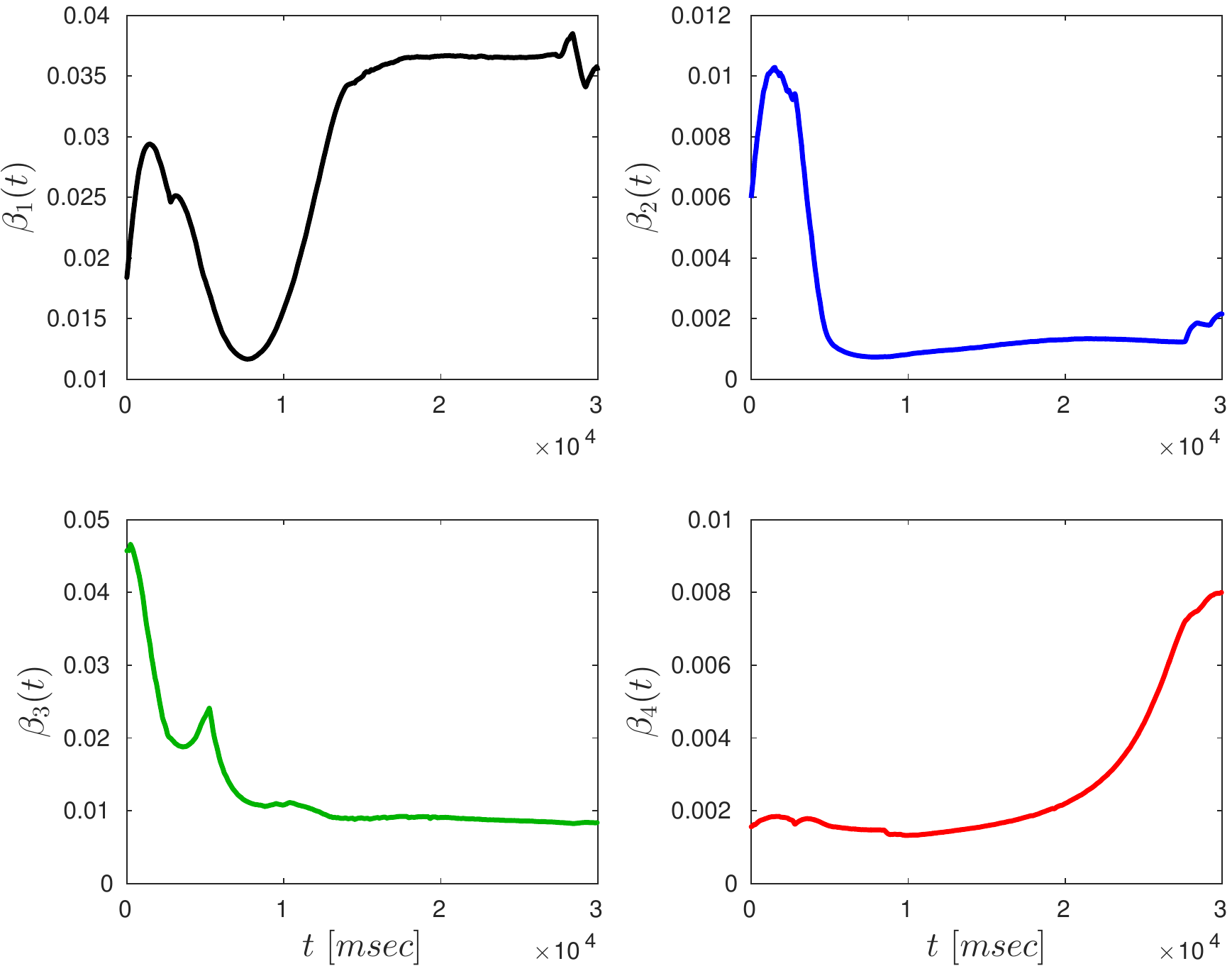}
	\caption{The evolution of the obstacle functions $\beta_i,\forall i\in\{1,\cdots,4\}$, which are shown to be always positive.}
	\label{fig:betas}       
\end{figure}


\section{Simulation Results} \label{sec:simulation_results}

To demonstrate the efficiency of the proposed control protocol, we consider a simulation example with $N=4$, $\mathcal{V}$ $= \{1,2,3,4\}$ spherical agents of the form \eqref{eq:system}, with $r_i$ $=0.25\text{m}$ and $d_i$ $= 5\text{m}$,$\forall i\in\{1,\dots,4\}$. The initial conditions are set to $p_1(0)=[-3,0,5]^\top \ \text{m}$, $p_2(0)=[-1,4,4]^\top \ \text{m}$, $p_3(0)=[-3,4,2]^\top \ \text{m}$, $p_4(0)=[-4,3,6]^\top  \ \text{m}$, $q_1(0)$ $=q_2(0)$ $=q_3(0)$ $=q_4(0)$ $= [0,0,0]^\top \ \text{r}$, which imply the initial neighboring sets $\mathcal{N}_1(0)=\{2\}$,$\mathcal{N}_2(0)=\{1,3,4\}$, $\mathcal{N}_3(0)=\{2,4\}$, $\mathcal{N}_4(0)=\{2,3\}$. The desired formation is defined by the feasible displacements $p_{12,\text{des}}= -p_{21,\text{des}} = [-1,-1,-2]^\top \ \text{m}, p_{23,\text{des}}= -p_{32,\text{des}} = [-2,-3,0]^\top \ \text{m}, p_{24,\text{des}}= -p_{42,\text{des}} = [-1,-2,0]^\top \ \text{m}, p_{34,\text{des}}= -p_{43,\text{des}} = [1,1,0]^\top \ \text{m}, q_{12,\text{des}}= -q_{21,\text{des}} = [-\tfrac{\pi}{4},0,-\tfrac{\pi}{4}]^\top \ \text{r}, q_{23,\text{des}}= -q_{32,\text{des}} = [-\tfrac{\pi}{12},0,0]^\top \ \text{m}, q_{24,\text{des}}= -q_{42,\text{des}} = [-\tfrac{\pi}{8},0,0]^\top \ \text{r}, q_{34,\text{des}}= -q_{43,\text{des}} = [\tfrac{5\pi}{24},0,0]^\top \ \text{r}$. We consider a workspace of radius $r_w = 10\text{m}$ containing two spherical static obstacles at $p_{o_1} = [-3,3,5]^\top\text{m}, p_{o_2} = [-1,1,3]^\top\text{m}$ with radii $r_{o_1} = r_{o_2} = 0.75\text{m}$. An illustration of the workspace with the agents at $t=0$ is given in Fig. \ref{fig:init_ws}. 
The potential function that we used is the one developed in \cite{dimarogonas2006feedback}.
The simulation results are depicted in Fig. \ref{fig:gammas}- Fig. \ref{fig:final_ws} for $t\in[0,30]\text{s}$. In particular, Fig. \ref{fig:gammas} shows the evolution of the goal functions $\gamma_i,\forall i\in\{1,\cdots,4\}$, which are decreasing to zero, whereas Fig. \ref{fig:betas} depicts the obstacle functions $\beta_i$, $\forall i\in\{1,\cdots,4\}$, which stay always positive. Furthermore, the control inputs are shown in Fig. \ref{fig:inputs}. Finally, the navigation of the agents in the workspace is pictured in Fig. \ref{fig:final_ws}. As proven in the theoretical analysis, the formation is successfully achieved and all the specifications of Problem \ref{problem} are met.

\begin{figure}[t!]
	\centering
	\includegraphics[trim={0 0 0cm -2cm},scale = 0.37]{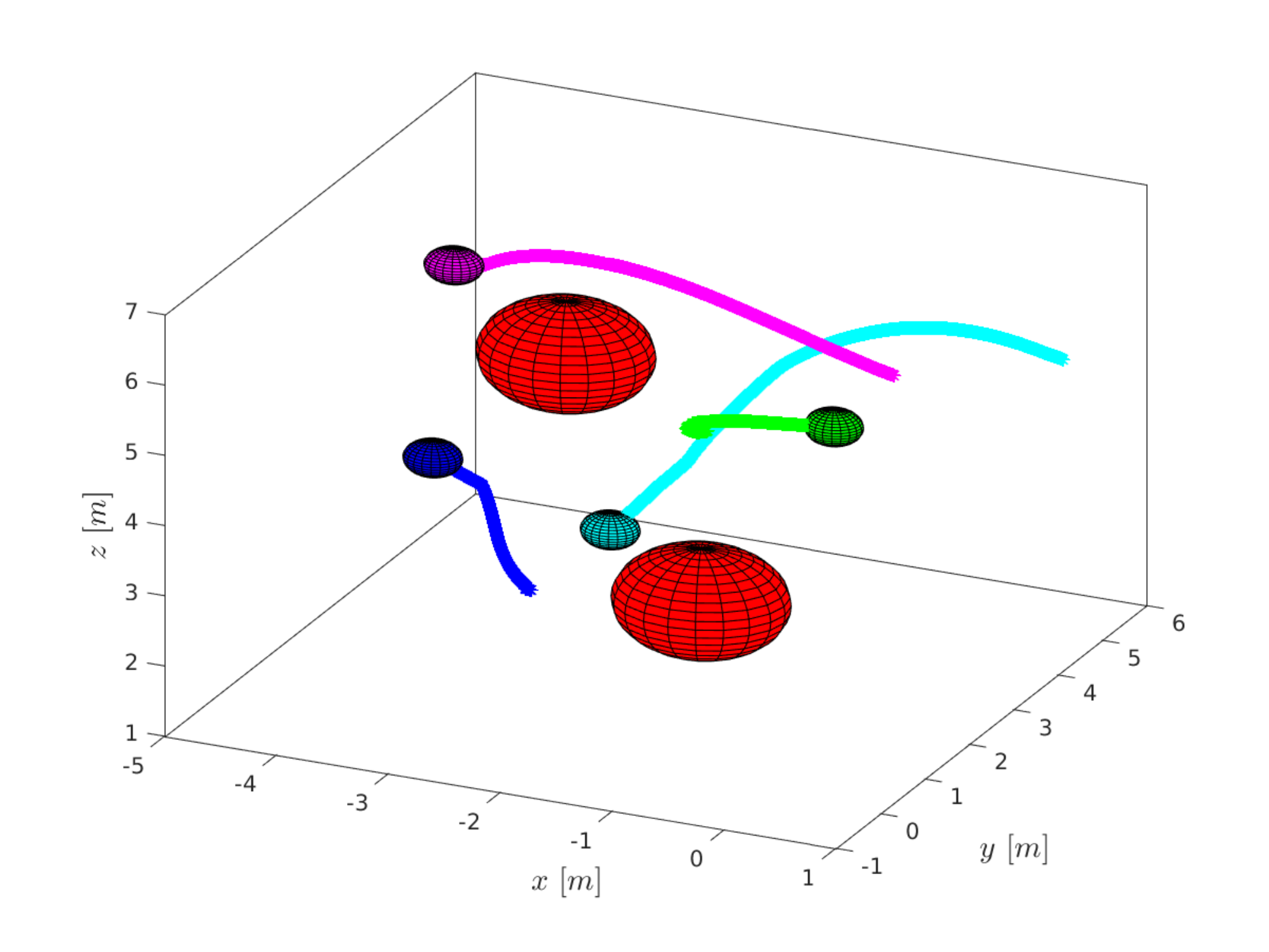}
	\caption{The motion of the agents in the workspace for $t\in[0,30]\text{s}$.}
	\label{fig:final_ws}       
\end{figure}

\section{Conclusions and Future Work} \label{sec:conclusions}

In this work we proposed a potential function- based decentralized control protocol for multi-agent systems which guarantees formation control with inter-agent collision avoidance, collision avoidance between the agents and the obstacles/workspace boundary, connectivity maintenance as well as singularity avoidance of multiple rigid bodies. Simulation results have  verified the validity of the proposed approach. Future efforts will be devoted towards developing global results as well as real-time experiments.

\bibliography{references}

\begin{thebibliography}{10}

\bibitem{oh_park_ahn_2015}
K.~Oh, M.~Park, and H.~Ahn, ``{A} {S}urvey of {M}ulti-{A}gent {F}ormation
  {C}ontrol,'' {\em Automatica}, vol.~53, pp.~424--440, 2015.

\bibitem{beard_2001_coordination}
R.~Beard, J.~Lawton, and F.~Hadaegh, ``{A} {C}oordination {A}rchitecture {F}or
  {S}pacecraft {F}ormation {C}ontrol,'' {\em IEEE TCST}, 2001.

\bibitem{egerstedt_formation}
M.~Egerstedt and X.~Hu, ``{F}ormation {C}onstrained {M}ulti-{A}gent
  {C}ontrol,'' {\em IEEE TRA}, vol.~17, no.~6, pp.~947--951, 2001.

\bibitem{fax_murray_2004}
L.~A. Fax and R.~Murray, ``{I}nformation {F}low and {C}ooperative {C}ontrol of
  {V}ehicle {F}ormations,'' {\em TAC}, vol.~49, no.~9, pp.~1465--1476, 2004.

\bibitem{do_2007_formation}
K.~Do, ``{B}ounded {C}ontrollers {F}or {F}ormation {S}tabilization of {M}obile
  {A}gents {W}ith {L}imited {S}ensing {R}anges,'' {\em TAC}, 2007.

\bibitem{dong_farrell_2008_cooperative}
W.~Dong and J.~Farrell, ``{C}ooperative {C}ontrol of {M}ultiple {N}onholonomic
  {M}obile {A}gents,'' {\em IEEE Transactions on Automatic Control}, vol.~53,
  no.~6, pp.~1434--1448, 2008.

\bibitem{anderson_yu_fidan_hendrickx_2008}
B.~Anderson, C.~Yu, B.~Fidan, and J.~Hendrickx, ``{R}igid {G}raph {C}ontrol
  {A}rchitectures for {A}utonomous {F}ormations,'' {\em IEEE CS}, 2008.

\bibitem{sepulchre_2008_symmeetric_formation}
D.~Paley, N.~Leonard, and R.~Sepulchre, ``{S}tabilization of {S}ymmetric
  {F}ormations to {M}otion {A}round {C}onvex {L}oops,'' {\em Systems and
  Control Letters}, vol.~57, no.~3, pp.~209--215, 2008.

\bibitem{krick_broucke_francis_2009}
L.~Krick, M.~Broucke, and B.~Francis, ``{S}tabilisation of {I}nfinitesimally
  {R}igid {F}ormations of {M}ulti-{R}obot {N}etworks,'' {\em IJC}, 2009.

\bibitem{dorfler_francis_2009}
F.~Dorfler and B.~Francis, ``Formation {C}ontrol of {A}utonomous {R}obots
  {B}ased on {C}ooperative {B}ehavior,'' {\em ECC}, pp.~2432--2437, 2009.

\bibitem{lin_jia_2010_rotating_formation}
P.~Lin and Y.~Jia, ``{D}istributed {R}otating {F}ormation {C}ontrol of
  {M}ulti-{A}gent {S}ystems,'' {\em SCL}, vol.~59, no.~10, pp.~587--595, 2010.

\bibitem{mesbahi_2010_graph_theory}
M.~Mesbahi and M.~Egerstedt, {\em {G}raph {T}heoretic {M}ethods in {M}ultiagent
  {N}etworks}.
\newblock Princeton University Press, 2010.

\bibitem{cao_morse_yu_anderson_dagsputa_2011}
M.~Cao, S.~Morse, C.~Yu, B.~Anderson, and S.~Dasgupta, ``{M}aintaining a
  {D}irected, {T}riangular {F}ormation of {M}obile {A}utonomous {A}gents,''
  {\em Communications in Information and Systems}, 2011.

\bibitem{belabbas2012robustness}
A.~Belabbas, S.~Mou, S.~Morse, and B.~Anderson, ``{R}obustness {I}ssues with
  {U}ndirected {F}ormations,'' {\em CDC}, pp.~1445--1450, 2012.

\bibitem{zavlanos_2008_distributed}
M.~Zavlanos and G.~Pappas, ``{D}istributed {C}onnectivity {C}ontrol of {M}obile
  {N}etworks,'' {\em TRO}, 2008.

\bibitem{basiri_2010_angle_formation}
M.~Basiri, A.~Bishop, and P.~Jensfelt, ``{D}istributed {C}ontrol of
  {T}riangular {F}ormations with {A}ngle-{O}nly {C}onstraints,'' {\em SCL},
  2010.

\bibitem{eren_2012_bearing_formation}
T.~Eren, ``{F}ormation {S}hape {C}ontrol {B}ased on {B}earing {R}igidity,''
  {\em International Journal of Control}, vol.~85, no.~9, pp.~1361--1379, 2012.

\bibitem{zhao2016bearing}
S.~Zhao and D.~Zelazo, ``{B}earing {R}igidity and {A}lmost {G}lobal
  {B}earing-{O}nly {F}ormation {S}tabilization,'' {\em TAC}, vol.~61, no.~5,
  2016.

\bibitem{oh_ahn_2014_angle_based_formation}
M.~Trinh, K.~Oh, and H.~Ahn, ``{A}ngle-{B}ased {C}ontrol of {D}irected
  {A}cyclic {F}ormations with {T}hree-{L}eaders,'' {\em ICMC}, 2014.

\bibitem{bishop_2015_distributed}
A.~N. Bishop, M.~Deghat, B.~Anderson, and Y.~Hong, ``{D}istributed {F}ormation
  {C}ontrol with {R}elaxed {M}otion {R}equirements,'' {\em IJRNC}, 2015.

\bibitem{fathian2016globally}
K.~Fathian, D.~Rachinskii, M.~Spong, and N.~Gans, ``{G}lobally {A}symptotically
  {S}table {D}istributed {C}ontrol for {D}istance and {B}earing {B}ased
  {M}ulti-{A}gent {F}ormations,'' {\em ACC}, pp.~4642--4648, 2016.

\bibitem{koditschek1990robot}
D.~Koditschek and E.~Rimon, ``{R}obot {N}avigation {F}unctions on {M}anifolds
  with {B}oundary,'' {\em Advances in Applied Mathematics}, 1990.

\bibitem{jadbabaie_nf_formation}
M.~Gennaro and A.~Jadbabaie, ``{F}ormation {C}ontrol for a {C}ooperative
  {M}ulti-{A}gent {S}ystem using {D}ecentralized {N}avigation {F}unctions,''
  {\em American Control Conference (ACC), 2006}, pp.~6--pp, 2006.

\bibitem{tanner_2005_formation_nf}
H.~Tanner and A.~Kumar, ``{F}ormation {S}tabilization of {M}ultiple {A}gents
  {U}sing {D}ecentralized {N}avigation {F}unctions.,'' {\em RSS}, vol.~1, 2005.

\bibitem{kan2012network}
Z.~Kan, A.~P. Dani, J.~M. Shea, and W.~E. Dixon, ``{N}etwork {C}onnectivity
  {P}reserving {F}ormation {S}tabilization and {O}bstacle {A}voidance via a
  {D}ecentralized {C}ontroller,'' {\em TAC}, vol.~57, no.~7, 2012.

\bibitem{dixon_2014_col_avoid_nf}
T.-H. Cheng, Z.~Kan, J.~A. Rosenfeld, and W.~E. Dixon, ``{D}ecentralized
  {F}ormation {C}ontrol with {C}onnectivity {M}aintenance and {C}ollision
  {A}voidance under {L}imited and {I}ntermittent {S}ensing,'' {\em ACC}, 2014.

\bibitem{dimarogonas2010analysis}
D.~Dimarogonas and E.~Frazzoli, ``{A}nalysis of {D}ecentralized {P}otential
  {F}ield {B}ased {M}ulti-{A}gent {N}avigation via {P}rimal-dual {L}yapunov
  {T}heory,'' {\em IEEE CDC}, pp.~1215--1220, 2010.

\bibitem{tanner2012multiagent}
H.~G. Tanner and A.~Boddu, ``Multiagent navigation functions revisited,'' {\em
  TRO}, vol.~28, no.~6, pp.~1346--1359, 2012.

\bibitem{dimos2005_nf_2nd_order}
D.~Dimarogonas and K.~Kyriakopoulos, ``{D}ecentralized {C}ontrol of {M}ultiple
  {A}gents with {D}ouble {I}ntegrator {D}ynamics,'' {\em IFAC}, 2005.

\bibitem{alex_chris_ppc_formation_ifac}
A.~Nikou, C.~Verginis, and D.~Dimarogonas, ``{R}obust {D}istance-{B}ased
  {F}ormation {C}ontrol of {M}ultiple {R}igid {B}odies with {O}rientation
  {A}lignment,'' {\em IFAC WC 2017, France, To appear}.

\bibitem{bechlioulis_tac_2008}
C.~Bechlioulis and G.~Rovithakis, ``{R}obust {A}daptive {C}ontrol of {F}eedback
  {L}inearizable {MIMO} {N}onlinear {S}ystems with {P}rescribed
  {P}erformance,'' {\em TAC}, vol.~53, no.~9, pp.~2090--2099, 2008.

\bibitem{Siciliano2009}
B.~Siciliano, L.~Sciavicco, and L.~Villani, {\em {R}obotics : {M}odelling,
  {P}lanning and {C}ontrol}.
\newblock Springer, 2009.

\bibitem{panagou2017distributed}
D.~Panagou, ``A distributed feedback motion planning protocol for multiple
  unicycle agents of different classes,'' {\em IEEE Transactions on Automatic
  Control}, vol.~62, no.~3, pp.~1178--1193, 2017.

\bibitem{dimarogonas2007decentralized}
D.~V. Dimarogonas and K.~J. Kyriakopoulos, ``Decentralized navigation functions
  for multiple robotic agents with limited sensing capabilities,'' {\em Journal
  of Intelligent \& Robotic Systems}, vol.~48, no.~3, pp.~411--433, 2007.

\bibitem{dimarogonas2006feedback}
D.~V. Dimarogonas, S.~G. Loizou, K.~J. Kyriakopoulos, and M.~M. Zavlanos, ``A
  feedback stabilization and collision avoidance scheme for multiple
  independent non-point agents,'' {\em Automatica}, vol.~42, no.~2,
  pp.~229--243, 2006.

\bibitem{koditschek1992robot}
E.~Rimon and D.~Koditschek, ``{E}xact {R}obot {N}avigation {U}sing {A}rtificial
  {P}otential {F}unctions,'' {\em IEEE TRA}, vol.~8, no.~5, pp.~501--518, 1992.

\end{thebibliography}
\bibliographystyle{ieeetr}
\addtolength{\textheight}{-12cm}

\end{document}